\documentclass[runningheads]{llncs}
\let\llncssubparagraph\subparagraph
\let\subparagraph\paragraph
\usepackage[compact]{titlesec}
\let\subparagraph\llncssubparagraph
\input{packages}

\begin{document}
\title{A HPC Co-Scheduler with Reinforcement Learning}





\author{%
Abel Souza\inst{1} \thanks{Research performed while working at Ume\aa\ University} $^{\textrm{\href{mailto:me@asouza.io}{\Letter}}}$  \and 
Kristiaan Pelckmans\inst{2} \and
Johan Tordsson\inst{3}
}%
\institute{
University of Massachusetts Amherst, USA\\
\email{asouza@cs.umass.edu}\\
\and
Uppsala University, Sweden\\
\email{kristiaan.pelckmans@it.uu.se}\\
\and
Ume{\aa} University, Sweden\\
\email{tordsson@cs.umu.se}
}


\maketitle
\begin{abstract}
Although High Performance Computing (HPC) users understand basic resource requirements such as the number of CPUs and memory limits,
internal infrastructural utilization data is exclusively leveraged by cluster operators, who use it to configure  batch schedulers. 
This task is challenging and increasingly complex due to ever larger cluster scales and heterogeneity of modern scientific workflows. As a result, HPC systems achieve low utilization with long job completion times (makespans).
To tackle these challenges, we propose a co-scheduling algorithm based on an adaptive reinforcement learning algorithm, where application profiling is combined with cluster monitoring. 
The resulting cluster scheduler matches resource utilization to application performance in a fine-grained manner (i.e., operating system level).
As opposed to nominal allocations, we apply decision trees to model applications' actual resource usage, which are used to estimate how much resource capacity from one allocation can be co-allocated to additional applications.
Our algorithm learns from incorrect co-scheduling decisions and adapts from changing environment conditions, and evaluates when such changes cause resource contention that impacts quality of service metrics such as jobs slowdowns.
We integrate our algorithm in an HPC resource manager that combines Slurm and Mesos for job scheduling and co-allocation, respectively.  
Our experimental evaluation performed in a dedicated cluster executing a mix of four real different scientific workflows demonstrates improvements on cluster utilization of up to 51\% even in high load scenarios, with 55\% average queue makespan reductions under low loads.

\end{abstract}

\section{Introduction}

High Performance Computing (HPC) datacenters process thousands of applications supporting scientific and business endeavours across all sectors of society. 
Modern applications are commonly characterized as data-intensive, demanding processing power and scheduling capabilities that are not well supported by large HPC systems~\cite{rodrigo2015a2l2}.
Data-intensive workflows can quickly change computational patterns, e.g., amount of input data at runtime. 
Unfortunately, traditional HPC schedulers like Slurm \cite{slurm02} do not offer Application Programming Interfaces (APIs) that allow users and operators to express these requirements. 
Current resource provisioning capabilities barely satisfy today's traditional non-malleable workloads, and as a result, most HPC centers report long queue waiting times~\cite{archersurvey19}, and low utilization~\cite{ambati2020waiting}.
These problems delay scientific outputs, besides triggering concerns related to HPC infrastructure usage, particularly in energy efficiency and in the gap between resource capacity and utilization~\cite{stevens2020ai}.
%

In contrast to HPC infrastructures that use batch systems and prioritize the overall job throughput at the expense of latency, cloud datacenters favor response time.
On the one hand, cloud resource managers such as Kubernetes/Borg~\cite{burns2016borg} and Mesos~\cite{2011mesos} assume workloads that change resource usage over time and that can scale up, down, or be migrated at runtime as needed~\cite{mell2011nist}.
This model allows cloud operators to reduce datacenter fragmentation and enable low latency scheduling while improving cluster capacity utilization.
On the other hand, HPC resource managers such as Slurm and Torque~\cite{torque06}, assume workloads with fixed makespan and constant resource demands throughout applications' lifespan, which must be acquired before jobs can start execution~\cite{reuther2018scalable, feitelson1996toward}. 
Consequently, to improve datacenter efficiency and to adapt to quick workload variations, resources should be instantly re-provisioned, pushing for new scheduling techniques. 

Thus, this paper proposes an algorithm with strong theoretical guarantees 
for co-scheduling batch jobs and resource sharing.
Considering the problem of a HPC datacenter where jobs need to be scheduled and cluster utilization needs to be optimized, we develop a Reinforcement Learning (RL) algorithm that models a co-scheduling policy that minimizes idle processing capacity and also respects jobs' Quality of Service (QoS) constraints, such as total runtimes and deadlines (Section \ref{sec:arch}).
Our co-scheduling algorithm is implemented by combining the Slurm batch scheduler and the Mesos dynamic job management framework.
The potential benefit of more assertive co-allocation schemes is substantial, considering that -- in terms of actual resource utilization, as opposed to nominal allocations -- the current HPC practice often leaves computational units idling for about 40\%~\cite{tirmazi2020borg, reiss2012heterogeneity}. 

We evaluate our solution by experiments in a real cluster configured with three different sizes and four real scientific workflows with different compute, memory, and I/O patterns. 
We compare our co-scheduling strategy with traditional space-sharing strategies that do not allow workload consolidation and with a static time-sharing strategy that equally multiplexes the access to resources.
Our RL co-scheduler matches applications QoS guarantees by using a practical algorithm that improves cluster utilization by up to 51\%, with 55\% reductions in queue makespans and low performance degradation (Section~\ref{sec:eval}).

\section{Background and Challenges}\label{sec:background}

The emergence of cloud computing and data-intensive processing created a new class of complex and dynamic applications that require low-latency scheduling and which are increasingly being deployed in HPC datacenters.
Low latency scheduling requires a different environment than batch processing that is dominant in HPC environments. 
HPC systems are usually managed by centralized batch schedulers that require users to describe allocations by the total run time (makespan) and amounts of resources, such as the number of CPUs, memory, and accelerators, e.g., GPUs~\cite{feitelson2014experience, reuther2018scalable, slurm02}.
As depicted in Figure~\ref{fig:batch_queue}, in the traditional HPC resource reservation model, known as space sharing, each job arrives and waits to be scheduled in a queue until enough resources that match the job's needs are available for use.
In this scenario, the batch system uses backfilling~\cite{lifka1995anlibm} to keep resources allocated and maximize throughput, 
and users benefit from having jobs buffered up (queued) and scheduled together, a technique known as gang scheduling~\cite{feitelson1996packing}. 

The main scheduling objectives in traditional HPC is performance predictability for parallel jobs, achieved at the expense of potential high cluster fragmentation,  scalability and support for low latency jobs.
In addition, this model assumes that the capacity of allocated computing resources is fully used, which is rarely the case~\cite{tirmazi2020borg}, especially at early stages of development.
This static configuration can be enhanced through consolidation, where jobs that require complementary resources share the same (or parts of) physical servers, which ultimately increases cluster utilization.
In HPC, consolidation happens mostly in non-dedicated resource components, such as the shared file systems and the network. 
Resource managers, such as Slurm, allow nodes to be shared among multiple jobs, but do not dynamically adjust the preemption time slices. 



\subsection{Resource Management with Reinforcement Learning}

Reinforcement learning (RL) is a mathematical optimization problem with states, actions, and rewards, 
where the objective is to maximize rewards~\cite{monahan1982state}.
In HPC, this can be formulated as a set of cluster resources (i.e., the environment) to where jobs need to be assigned by following an objective function (i.e., the rewards). 
Three common objective functions in HPC are minimizing expected queue waiting times, guaranteeing fairness among users, and increasing the cluster utilization and density.
The scheduler's role is to model the cluster -- i.e., to map applications behaviour to provisioned resources -- through actionable interactions with the runtime system, such as adjusting cgroup limits~\cite{cgroup}.
Finally, to evaluate actions, the scheduler can observe state transitions in the cluster and subsequently calculate their outcomes to obtain the environment's reward. 
One example of action used extensively in this work is to co-schedule two applications A and B onto the same server and measure the reward (or loss) in terms of runtime performance. 
The latter compared to a scenario with exclusive node access for A and B. 
In this RL framework, the scheduler is the learner and interactions allow it to model the environment, i.e., the quantitative learning of the mapping of actions to outcomes observed by co-scheduling applications.
A reward (or objective) function describes how an agent (i.e., a scheduler) should behave, and works as a normative evaluation of what it has to accomplish. 
There are no strong restrictions regarding the characteristics of the reward function, but if it quantifies observed behaviors well, the agent learns fast and accurately.

\subsection{Challenges}
Common HPC resource managers do not profile jobs nor consider job performance models during scheduling decisions.
Initiatives to use more flexible policies are commonly motivated by highly dynamic future exascale applications, where runtime systems need to handle orchestration due to the large number of application tasks spawned at execution time~\cite{zhang2019rlscheduler, castain2018pmix}.
Currently, dynamic scheduling in HPC is hindered as jobs come with several constraints due to their tight coupling, including the need for periodic message passing for synchronization and checkpointing for fault-tolerance~\cite{gainaru2019speculative}.
Characteristics such as low capacity utilization \cite{ambati2020waiting} show a potential to improve cluster efficiency, where idle resources -- ineffectively used otherwise -- can be allocated to other applications with opposite profile characteristics.

Thus, different extensions to the main scheduling scheme are of direct relevance in practical scenarios.
Server \textit{capacity} can be viewed as a single-dimensional variable, though in practice it is multivariate and includes CPU, memory, network, I/O, etc.
As such, a practical and natural extension to this problem is to formulate it as a multi-dimensional metric \cite{li2014dynamic}.
However, co-scheduling one or more jobs affects their performance as a whole \cite{janus2017slo}, meaning that the actual capacity utilization of a job depends on how collocated jobs behave, making this capacity problem even more complex.

\begin{figure}[!ht]
    \centering
    \subfigure[Batch Architecture]
    {
        \includegraphics[width=0.7\textwidth]{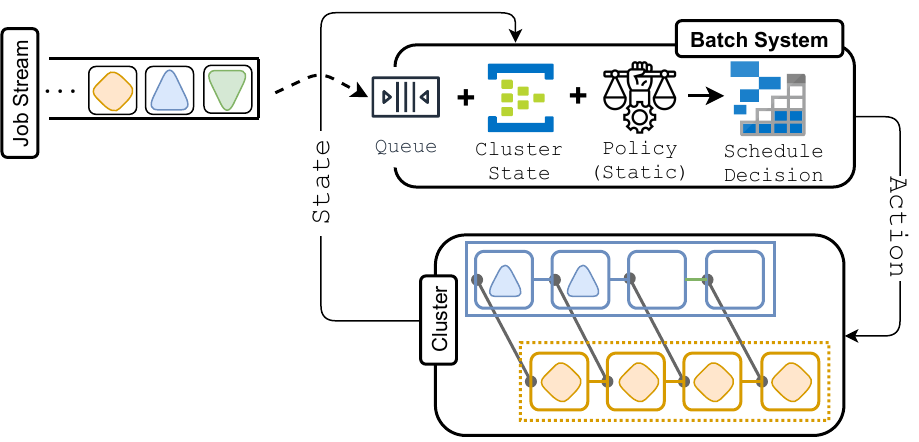}
        \label{fig:batch_queue}
    }
    \subfigure[ASA$_X$ Architecture]
    {
        \includegraphics[width=0.7\textwidth]{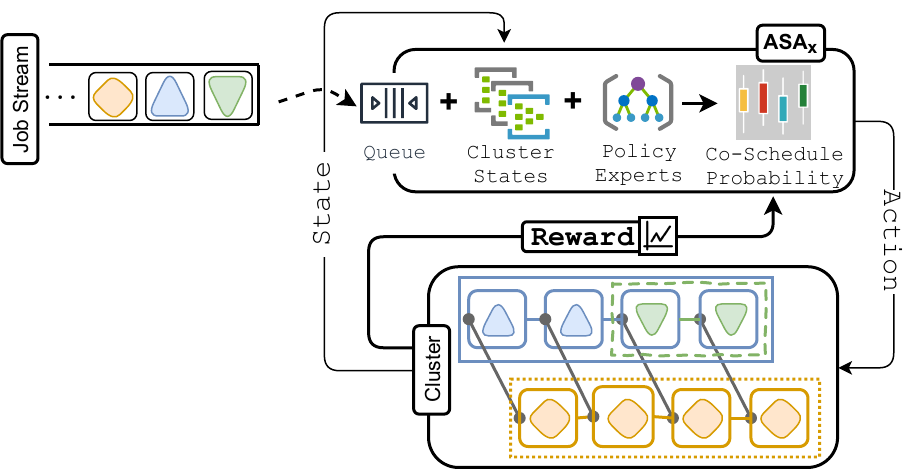}
        \label{fig:asax_architecture}
    }
    \caption{(a) Batch and (b) ASA$_X$ Architectures: (a) A traditional batch system such as Slurm. The Upside down triangle (green) job waits for resources although the cluster is not fully utilized; (b) ASA$_X$, where rewards follow co-scheduling decision actions, steered through \textit{Policy Experts}. In this example, the upside down triangle job (green) is co-allocated with the triangle job (blue).}
\end{figure}

\subsection{The Adaptive Scheduling Architecture} \label{subsec:asa}

ASA -- The Adaptive Scheduling Architecture is an architecture and RL algorithm that reduces the user-perceived waiting times by accurately estimating queue waiting times, as well as optimizes scientific workflows resource and cluster usage~\cite{asa}.
ASA encapsulates application processes into containers and enable fine-grained control of resources through and across job allocations.
From user-space, ASA enables novel workflow scheduling strategies, with support for placement, resource isolation and control, fault tolerance, and elasticity.
As a scheduling algorithm, ASA performs best in settings where similar jobs are queued repeatedly as is common setting in HPC. 
However, ASA's inability to handle states limits its use in broader scheduling settings.
As this paper shows, embedding states such as application performance, users' fair-share priorities, and the cluster capacity directly into the model enables more accurate scheduling decisions.

\section{A Co-Scheduler Architecture and Algorithm}\label{sec:arch}

In this section we introduce ASA$_X$: the HPC co-scheduler architecture and its algorithm. 
The goals of ASA$_X$ are higher cluster throughput and utilization, while also controlling performance degradation as measured by application completion time.
Additionally, ASA$_X$ introduces the concept of cluster and application states (Section~\ref{sec:arch}).
When measured, these concepts influence how each co-scheduling action is taken.
Handling states is the main difference to ASA and enables the creation of scheduling strategies that achieve an intelligent exploitation (finding the Pareto frontier) of the space spanned by the application QoS requirements and the available compute capacity.
As such, evaluating ASA$_X$'s performance depends on how time-sharing strategies compare to space-sharing strategies, the default policy of most HPC clusters that
guarantees predictable performance.
We handle the problem of growing number of actions by incorporating expert decision tree structures that can be computed efficiently.
Together with the loss functions, decision trees can be used efficiently, and easily be updated to modify the probabilities that affect how co-scheduling actions are chosen and how resources are used.

\subsection{Architecture and Algorithm Overview}\label{sec:arch}


Fig. \ref{fig:asax_architecture} shows ASA$_X$'s architecture, where jobs are queued as in regular HPC environments, but allocated resources are shared with other jobs. 
The main differences from regular batch scheduling (illustrated in Fig. \ref{fig:batch_queue}) is that a single deterministic cluster policy is replaced with a probabilistic policy that evolves as co-scheduling decision outcomes are evaluated.
The job co-scheduling is implemented by an algorithm (see Section~\ref{sec:alg_details}) that creates a collocation probability distribution by combining the cluster and application states, the cluster policy, and the accumulated rewards. 
Note that the cluster 'Policy' is analytically evaluated through decision trees (DTs) (Fig. \ref{fig:dt}), which correspond to the (human) \textit{experts} role in RL.
Experts map the applications and system online features, represented by internal cluster metrics such as CPU and memory utilization, into collocation decisions.
Combined through multiple DTs, these metrics form the cluster \textit{state}.
Each leaf in the DT outputs a multi-variate distribution p$_{i,j}$ that represents the likelihood of taking a specific action, and combined they affect how co-scheduling decisions are made.
In here, actions are defined by how much CPU cgroup \cite{cgroup} quota each job receives from the resources allocated to them, which influence how the operating system scheduler schedules application processes at runtime.
Another difference with traditional HPC architectures is the addition of a feedback mechanism, where \textit{rewards} are accumulated after each co-scheduling action.
This mechanism allows the architecture to asses the outcome of a co-scheduling decision to improve future decisions.
To enable a job-aware scheduler, decisions are implemented in a per-job basis (see Section~\ref{sec:alg_details}).
Thus, ASA$_X$ aims at discovering which expert is the best to minimize the performance degradation faced by a job due to co-scheduling decisions.

There can be an arbitrary number of experts, and they can be described by anything that is reasonably lean to compute, from functions to decision trees.
However, the combination of all experts needs to approximate the current cluster state as precisely as possible.
To define and compute experts, in this paper we use metrics such as CPU ({\fontfamily{qcr}\selectfont CPU\%}) and memory utilization ({\fontfamily{qcr}\selectfont Mem\%}), workflow stage {\fontfamily{qcr}\selectfont type} (i.e. \textit{sequential}, where only one core is utilized, or \textit{parallel}, where all cores are utilized), time {\fontfamily{qcr}\selectfont interval} since the job started execution related to its duration (e.g. 25\%, 50\%, 75\%). 
In addition to these metrics, we also define {\fontfamily{qcr}\selectfont Hp(t)}, a {\fontfamily{qcr}\selectfont happiness} metric that is defined at time $t$ for a given job as

\begin{equation}
{\fontfamily{qcr}\selectfont Hp(t)} = \frac{|t_{Walltime} - t| * \# Remaining Tasks}{\#Tasks/s}.
\label{eq.happiness}
\end{equation}
Devised from a similar concept \cite{lakew2015performance}, the happiness metric relates the job's remaining execution time with its remaining amount of work.
It enables the analytical evaluation of a running job in regards to the allocated resource capacity and its overall throughput.
Given the remaining time ($|t_{Walltime} - t|$) and performance ($\#Tasks/s$), if {\fontfamily{qcr}\selectfont Hp(t)} is near to, or greater than $1.0$, then it can be inferred that the job is likely to complete execution within the wall-time limit; else (if {\fontfamily{qcr}\selectfont Hp(t)} $< 1.0$), the job is likely to not complete by the wall-time.
It is important to note that the {\fontfamily{qcr}\selectfont Hp(t)} metric best describes jobs that encapsulate one application process.
 This is achieved by encapsulating each workflow stage into a single job, as is common in scientific workflows (see Section \ref{sec:eval}).

\begin{figure}[t] \centering
	\includegraphics[width=0.85\columnwidth]{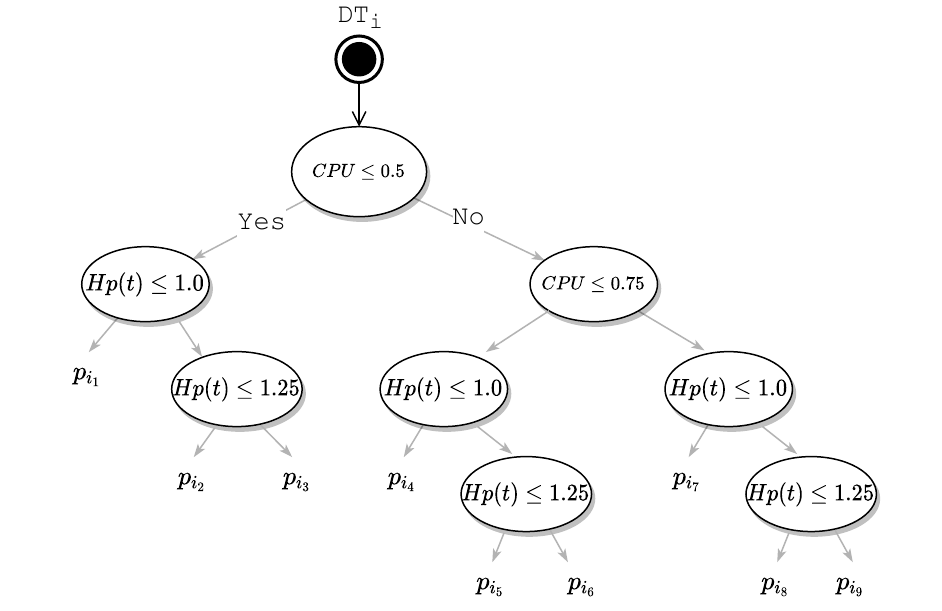}
	\DeclareGraphicsExtensions.
	\caption{A Decision tree (DT) expert structure illustrating the evaluation of the 'CPU' state in an allocation. For each DT 'Metric ({\fontfamily{qcr}\selectfont CPU\%}, {\fontfamily{qcr}\selectfont Mem\%}, {\fontfamily{qcr}\selectfont type}, and {\fontfamily{qcr}\selectfont interval}), an action strategy is devised by combining four different distributions $p_{i}$. Then, depending on the state for each 'Metric', one distribution among nine ($p_{i_1}, ..., p_{i_9}$) is returned. }
	\label{fig:dt}
\end{figure}

\subsubsection{Cluster Policy.}

We combine states through DTs for each metric with the {\fontfamily{qcr}\selectfont happiness} metric.
Figure \ref{fig:dt} illustrates one possible expert for the CPU state.
It evaluates if the {\fontfamily{qcr}\selectfont CPU\%} is high (e.g. $> 0.75$), and if so, and if {\fontfamily{qcr}\selectfont Hp(t)} is near to $1.0$,  
the policy expects high performance degradation due to job collocation.
Conversely, if the {\fontfamily{qcr}\selectfont CPU\%} is low, and the {\fontfamily{qcr}\selectfont Hp(t)} is greater than $1.1$, the performance degradation is likely to be small, and the expert increases the probability for jobs to be collocated.
Rather than describing these relations explicitly, we work with a mixture of different decision trees. 
Hence, $\pv_{i,j}(\xv)$ represents how likely an action $a_j$ is to be taken given the evaluated state $x_t$, according to the $i$-th expert in DT$_i$.
For instance, if there are four CPU quota co-allocations e.g. (0\%, 25\%, 50\%, 75\%), each degrading {\fontfamily{qcr}\selectfont Hp(t)} differently,
 {\fontfamily{qcr}\selectfont DT1} may evaluate how state $x_t$ is mapped into a distribution of performance degradation $\pv_{i,j}(\xv)$ for a situation where {\fontfamily{qcr}\selectfont CPU\% > 0.75}, and {\fontfamily{qcr}\selectfont Hp(t) $\approx$ 1.0}. 
This scenario could return $\pv_{1}(\xv) = (0.6, 0.3, 0.05, 0.05)$, meaning action $a_0$ (= 0\%), i.e. no co-allocation, is likely the best co-schedule decision.
Another DT2 could evaluate another state, e.g. related to memory or to how a given special resource behaves, etc., returning a different $\pv_{i,j}(\xv)$, with a different co-allocation action, impacting the final $\pv$ accordingly (line 6 in Algorithm \ref{alg.asax}).
In general, we design our decision trees to suggest more collocation when $Hp(t) \geq 1.0$, and ASA$_X$ looks for other ongoing allocations whenever a 0\% co-allocation is the best decision for a certain running job.

\subsection{Algorithm Details}\label{sec:alg_details}

\begin{table}
\caption{Algorithm variables and their descriptions.}
\label{tab:alg_desc}
\centering
\resizebox{0.9\textwidth}{!}{%
\begin{tabular}{lll}
\midrule
\textbf{Feature} & \textbf{Name}                                              & \textbf{Description} \\ \midrule
\textit{\{a\}}   & Action set                                                 & \begin{tabular}[c]{@{}l@{}}Actions that users or operators can take while\\scheduling or allocating resources to applications.\end{tabular} \\ \hline
\textit{DT$_i$}   & \begin{tabular}[c]{@{}l@{}}Decision Tree\end{tabular}    & \begin{tabular}[c]{@{}l@{}}Used to assess the weight a given application\\or resource metric value ought to influence actions.\end{tabular}  \\\hline
R$_i$             & Risk                                                       & Risk for using DT$_i$. \\ \hline
r$_i$             & \begin{tabular}[c]{@{}l@{}}Accumulated\\Risk\end{tabular} & \begin{tabular}[c]{@{}l@{}}Accumulation of all risks incurred by DT$_i$.\end{tabular} \\ \hline
$\ell$(a)             & \begin{tabular}[c]{@{}l@{}}Loss function\end{tabular}    & \begin{tabular}[c]{@{}l@{}}The performance degradation due to collocation.\end{tabular}                                                      \\ \midrule
\end{tabular}
}
\end{table}

In here we describe the internal details of how ASA$_X$ is implemented, with variable and details presented in Table \ref{tab:alg_desc}. 
A co-scheduling action depends on the associated state $\xv\in\R^d$, where $d$ depends on the number of selected cluster metrics as shown in Fig. \ref{fig:dt}.
A state $\xv$ is a vector that combines the assessment of all distributions output by the DTs.
At any time $t$ that a co-scheduling decision is to be made, one action $a$ out of
the $m$ possible actions $a = \{a_0, a_1, \dots, a_{m-1}\}$ is taken with the system in state $\xv$.
However, while in state $\xv$, co-allocating resources from a running Job$_a$ to a Job$_b$ may degrade Job$_a$'s performance, and this incurs in an associated loss $\ell(a)$.
This overall idea is represented in Algorithm \ref{alg.asax}, where we have a nested loop.
The outer loop (starting in line 1) initiates and updates the parameters for each scheduling decision that can be made.
Once an ongoing job allocation candidate is presented, the inner loop (line 3) co-allocates resources to other jobs until the cumulative loss exceeds a threshold, in which case learning needs to be updated.
In mathematical terms, let $\pv:\R^d\rightarrow\R^m$ be a function of states $\xv\in\R^d$ to co-allocation actions, and the goal of learning is to optimally approximate this mathematical relationship.
The central quantity steering $\pv$ is the {\em risk} $\RM$, here defined for each DT$_i$ as a vector in $\R^n$ where 
\begin{equation}
	\RM_i = \ell(a) \pv_i(a).
	\label{eq.asax.}
\end{equation}
To quantify the performance degradation due to co-scheduling at iteration $t$, we maintain a vector $\rv_t$ of the total accumulated risk (line 10).
At each co-scheduling decision at time $t$, we let $\rv_{t,i}$ denote the accumulated risk of the $i$th expert.
ASA$_X$ then creates a strategy to minimize $\rv_i$: 

\begin{equation}
	\min_\ast \sum_{s=1}^t \sum_{j=1}^{n_s} \pv_\ast(\xv_{sj}) \ell\left(a_{sj}\right).
	\label{eq.excess}
\end{equation}

Note that the strategy minimising $\rv_t$ corresponds to the maximum likelihood estimate, meaning we want to follow the expert that assigns the highest probability to the actions with the lowest loss.
In here $m$ actions $a = \{a_0, a_1, \dots, a_{m-1}\}$ correspond to a discretization of resource allocations, and are expressed in terms of CPU\% quota allocation ratios of $m$ intervals, e.g. if $m=10$ then $a = \{a_{0}=0\%, a_{1}=10\%, \dots, a_{9} = 90\%\}$.
This means a fraction $a_m$ of an ongoing $Job_a$ allocation is re-allocated to a $Job_b$.
When co-scheduling $a_m$ from an allocation to other job, we calculate the performance degradation loss $\ell(a)$, limiting it to $1.0$ and proportionally to the actual job submission execution time, normalized by its wall-time.
For example, if an user requests a wall-time of 100s, and the submission finishes in $140s$, the loss would be proportional to $min(1.0, |140-100|/100)$.
On the other hand, if the submission finishes in 70s, then the loss is $0$, because it successfully completed within its wall-time limits.
ASA$_X$ then optimizes $\pv$ over the $m$ quota allocations (actions) according to the loss $\ell(a)$ and the accumulated risks $\rv_{ti}$ for $Job_a$.
Accumulated risks help in cases where the co-scheduling strategy incurs in high performance degradation.
Therefore, in Algorithm \ref{alg.asax} we set a threshold $\rv_t$ (line 3) bounding the accumulated risk for a given $Job_a$, and reset it in the situation when the co-location actions resulted in large performance degradation losses (line 12).
This allows Algorithm \ref{alg.asax} to update its knowledge about $Job_a$ and bound the accumulated mistakes.
Notably, we prove that the {\em excess risk} $E_t$ after $t$ co-scheduling decisions is bound, meaning the algorithm converges to an optimal point in a finite time after a few iterations 
(see Appendix~\ref{sec:appendix}).


\begin{algorithm}[!ht]
\caption{ASA$_{\mbox{x}}$}
\label{alg.asax}
\begin{algorithmic}[1]
\REQUIRE \ \\Queued $Job_b$\\ $Job_a$ allocation satisfying $Job_b$ {\fontfamily{qcr}\selectfont \textbf{\#}In terms of resource}\\$m$ co-allocation actions $a$, e.g. $m=10$ and $a = \{a_0=0\%,...,a_{m-1}=90\%\}$ \\Initialise $\alpha_{0i}=\frac{1}{n}$ for $i=1, \dots,n$ {\fontfamily{qcr}\selectfont \textbf{\#}metrics in state $\xv$, e.g. CPU, Mem.}\\ \
\FOR{$t=1,2, \dots$}
\STATE Initialise co-allocation risk $\rv_{ti}=0$ for each $i$-th metric\\
{\fontfamily{qcr}\selectfont \textbf{\#}Co-allocation assessment:}
\WHILE{$\max_i \rv_{ti} \leq 1$}
\STATE Evaluate $Job_a$'s state $\xv$:
\STATE \hspace{.4cm}Compute each $i$th-DT metric expert $\pv_i(\xv)$ $\in \R^m$
\STATE \hspace{.4cm}Aggregate $Job_a$'s co-schedule probability as $\pv=\sum_{i=1}^n \alpha_{t-1,i} \pv_i(\xv) \in\R^m$
\STATE $j$ = sample one action from $a$ according to $\pv$
\STATE Allocate $a_j$ from $Job_a$'s to co-schedule $Job_b$ 
\STATE Compute $Job_a$'s performance loss $|\ell(a_j)| \leq 1$ due to co-allocation
\STATE For all $i$, update $Job_a$'s risk $\rv_{ti} = \rv_{ti} + \pv_i(a_j) \ell\left(a_j\right)$
\ENDWHILE
\STATE For $Job_a$ and for all $i$, update $\alpha_{t,i}$ as
	\[ \alpha_{t,i} = \frac{\alpha_{t-1,i}}{N_t} \times e^{- \gamma_t \rv_{ti}}\]
	where $N_t$ is a normalising factor such that $\sum_{i=1}^n \alpha_{t,i} = 1$.
\ENDFOR
\end{algorithmic}
\end{algorithm}

Finally, rigid jobs can only start execution when all requested resource are provisioned (see Section \ref{sec:background}). 
To simultaneously achieve this and to improve cluster utilization, HPC infrastructures use \textit{gang scheduling} \cite{feitelson1996packing}.
By running on top of Mesos~\cite{2011mesos} (see Section \ref{subsec:asa}), ASA$_X$ can scale to manage thousands of distributed resources. 
These features include fine-grained resource isolation and control (through cgroup quotas), fault tolerance, elasticity, and support for new scheduling strategies.
However, although Mesos supports resource negotiation, it does not support gang scheduling, nor timely reservation scheduling.
Therefore, ASA$_X$ also works as a framework to negotiate resource offers from Mesos in a way that satisfies queued jobs' requirements and constraints.
ASA$_X$ withholds offers, allocating them when enough resources fulfill the job's requirements, and releasing the remaining ones for use by other allocations.
This is a basic version of the \textit{First fit} capacity based algorithm~\cite{feitelson1996packing}.
Our focus is towards improving cluster throughput and utilization, which may have direct impacts on reducing queue waiting times and jobs response times.
Although in general many other aspects such as data movement and locality also influence parallel job scheduling, we do not (explicitly) address them. 

\section{Evaluation} \label{sec:eval}

In this section we evaluate ASA$_X$ with respect to workloads' total makespan, cluster resource usage, and workflows total runtime.
We compare ASA$_X$ against a default backfilling setting of Slurm and a static co-allocation setting.

\subsection{Computing System}

The experimental evaluation is performed on a system with high performance networks, namely a NumaScale system~\cite{rustad2013numascale} with 6 nodes, each having two AMD Opteron Processor (6380) 24-cores and 185 GB memory.
The NumaConnect hardware interconnects these nodes and appears to users as one large \textit{single} server, with a total of 288 cores and 1.11 TB of memory.
Memory coherence is guaranteed at the hardware level and totally transparent to users, applications, and the OS itself.
Servers are interconnected through a switch fabric 2D Torus network which supports sub microsecond latency accesses.
The NumaScale storage uses a XFS file system, providing 512 GB of storage.
The system runs a CentOS 7 (Kernel 4.18) and jobs are managed by Slurm (18.08) with its default backfilling setting enabled. 
When managing resources in the static and ASA$_X$ settings, jobs are managed by Mesos 1.9 that uses our own framework/co-scheduler on top of it (see Fig. \ref{fig:asax_architecture} and Section \ref{sec:arch}).

\subsection{Applications}


Four scientific workflows with different characteristics were selected for our evaluation Montage, BLAST, Statistics, and Synthetic.
\\
\\
\noindent\textbf{Montage} \cite{berriman2004montage} is an I/O intensive application that constructs the mosaic of a sky survey. 
The workflow has nine stages, grouped into two parallel and two sequential stages.
All runs of Montage construct an image survey from the {\fontfamily{qcr}\selectfont 2mass} Atlas images.
\\
\\
\noindent\textbf{BLAST} \cite{altschul1997gapped} is a compute intensive applications comparing DNA strips against a database larger than 6 GB. 
It maps an input file into many smaller files and then reduces the tasks to compare the input against the large sequence database. 
BLAST is composed of two main stages: one parallel followed by one sequential.
\\
\\
\noindent \textbf{Statistics} is an I/O and network intensive application which calculates various statistical metrics from a dataset with measurements of electric power consumption in a household with an one-minute sampling rate over a period of almost four years \cite{kolter2011redd}. 
The statistics workflow is composed mainly of two stages composed of two sequential and two parallel sub-stages. The majority of its execution time is spent exchanging and processing messages among parallel tasks.
\\
\\
\noindent \textbf{Synthetic} is a two stages workflow composed of a sequential and a parallel task.
This workflow is both data and compute intensive. 
It first loads the memory with over one billion floating point numbers (sequential stage), and then performs additions and multiplications on them (parallel stage).

\subsection{Metrics}

The following metrics are used in the evaluation.
The \emph{total runtime} is measured by summing up the execution times for each workflow stage, submitted as separate jobs. Equally important, \emph{the response time} (also known as makespan, or flow time) is defined as the time elapsed between the submission and the time when the job finishes execution.
A related metric is the waiting time, which is the time one job waits in the queue before starting execution. Additional metrics are CPU and memory utilization as measured by the Linux kernel.
These latter two capture the overall resource utilization, and aids understanding of how well  co-schedulers such as ASA$_X$ model application performance. 


\subsection{Workloads}

To evaluate ASA$_X$ co-scheduling we compare it against a Static CPU quota configuration and a default (space-sharing) Slurm setting with backfilling enabled.
We execute the same set of 15 workflow jobs (4 Montage, 4 BLAST and 4 Statistics [8, 16, 32, 64 cores], and 3 Synthetic [16, 32, 64 cores]) three times, each with three size configurations: namely 64 (x2), 128 (x4), and 256 (x8) cores.
The workflows have different job geometry scaling requests ranging from 8 cores to up to 64 cores (smallest cluster size, x2), totalling 512 cores and 45 job submissions for each cluster size.
This workload selection demonstrates how the different scheduling strategies handle high (cluster size x8 = 256 cores), medium (x4 = 128 cores), and low loads (x2 = 64 cores), respectively.
When comparing the Static setting and ASA$_X$ against Slurm, neither of them get access to more resources (i.e. cores) than the cluster size for each experimental run.
In all experiments, Slurm statically allocates resources for the whole job duration. 
Moreover, when scheduling jobs in the Static configuration, collocations of two jobs in a same server are allowed and the time-sharing CPU quota ratio is set to 50\% for each job (through cgroups \cite{cgroup}). 
Collocation is also done for ASA$_X$, but the CPU quotas are dynamically set and updated once the rewards are collected according to Algorithm \ref{alg.asax}.
Notably, to reduce scheduling complexity, our First Fit algorithm does not use Mesos resource offers coming from multiple job allocations.
Finally, as mentioned in the previous section, the loss function $\ell\left(a\right)$ to optimize the co-schedule actions is calculated proportionally to the user requested wall-time and to the actual workflow runtime, similarly to the {\fontfamily{qcr}\selectfont Hp(t)} metric.
The $\ell\left(a\right)$ values span 0.0 and 1.0, where 1.0 means performance degraded by at least 50\%, and 0.0 means no performance degradation.

\subsection{Results}

Figures \ref{fig:makespan} and \ref{fig:runtime}, and Tables \ref{tab:results} and \ref{tab:summary} summarize all experimental evaluation, which is discussed below.

\subsubsection{Makespan and Runtimes.}

Figures \ref{fig:makespan} and \ref{fig:runtime} show respectively queue workload makespans and workflow runtimes, both in hours.
We note that the total makespan reduces as the cluster size is increased.
As expected, Slurm experiments  have small standard deviations as they use the space-sharing policy and isolate jobs with exclusive resource access throughout their lifespan.
Static allocation yields longer makespans as this method does not consider actual resource usage by applications, but rather takes deterministic decisions about co-scheduling decisions, i.e. allocates half the CPU capacity to each collocated application.
In contrast, ASA$_X$, reduces the overall workload makespan by up to 12\% (64 cores) as it learns overtime which jobs do not compete for resources when co-allocated.
For this reason, some initial ASA$_X$ scheduling decisions are incorrect, which explains the higher standard deviations shown in Fig. \ref{fig:makespan}.

Table \ref{tab:results} and Figure \ref{fig:runtime} show the aggregated average runtimes for each workflow, with their respective standard deviations.
Notably, the high standard deviations illustrate the scalability of each workflow, i.e. the higher standard deviation, the more scalable the workflow is (given the same input, as it is the case in the experiments).
Figure \ref{fig:runtime}  shows different runs for the same application with different number of cores. 
If an application is scalable, the more cores the allocation gets, the faster the application completes execution, and the larger the standard deviation. 
Conversely, when the application is not scalable, performance is not improved when more cores are added to the allocation. 
The application cannot take advantage of extra resources, resulting in more idle resources, thus reducing the standard deviation between the runs.
Finally, BLAST and Synthetic are two very scalable, CPU intensive workloads, which do not depend on I/O and network as Montage and Statistics do.
The Static strategy has the highest standard deviations overall because its workload experiences more performance degradation when compared to both Slurm and ASA$_X$, which this is due to its static capacity allocation.

\begin{figure*}[!ht]
    \centering
    \subfigure[Total workload queue makespan per strategy and cluster size]
    {
        \includegraphics[width=0.8\linewidth]{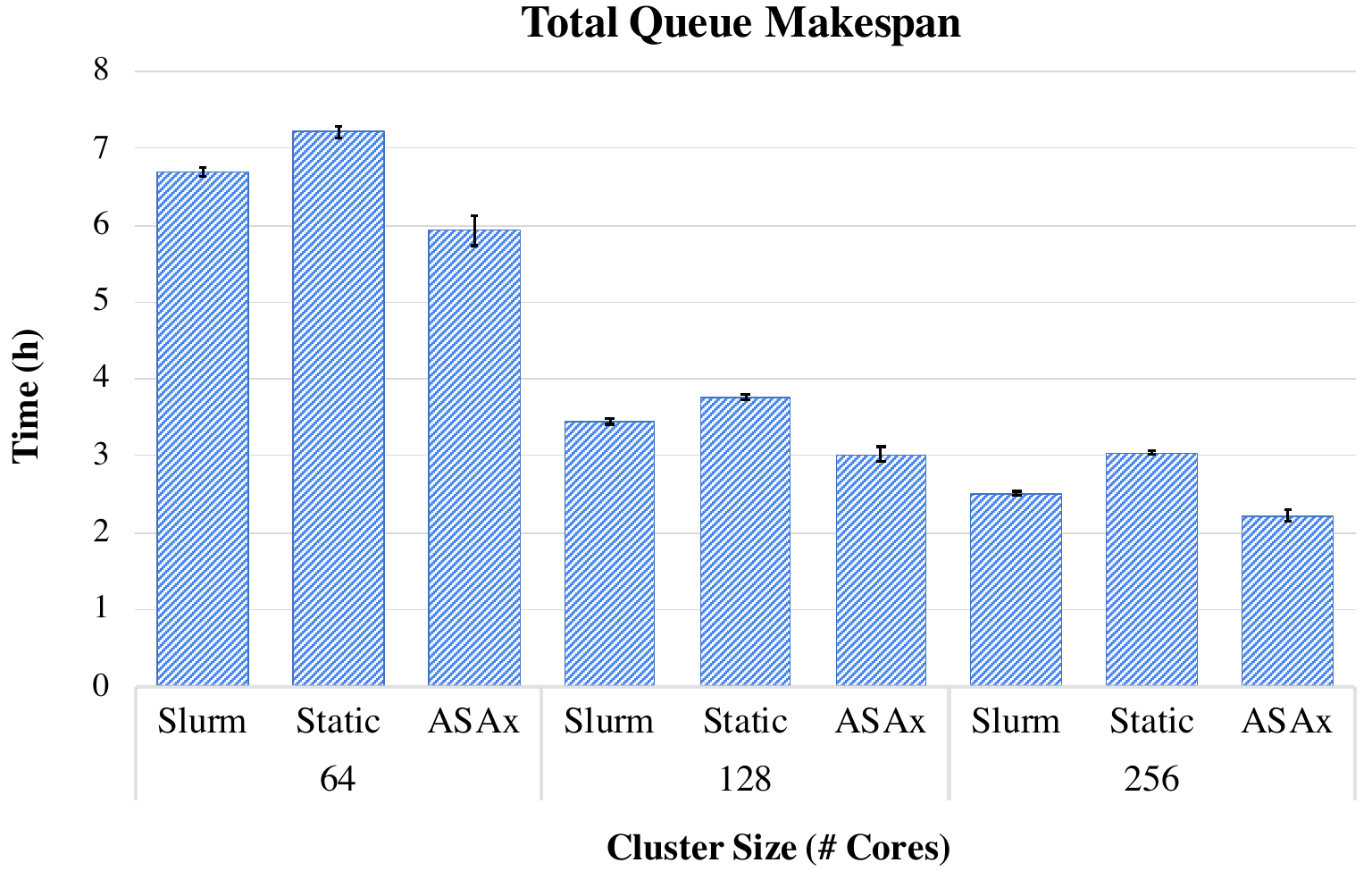}
        \label{fig:makespan}
    }
    \subfigure[Average total runtime per strategy and cluster size]
    {
        \includegraphics[width=0.8\linewidth]{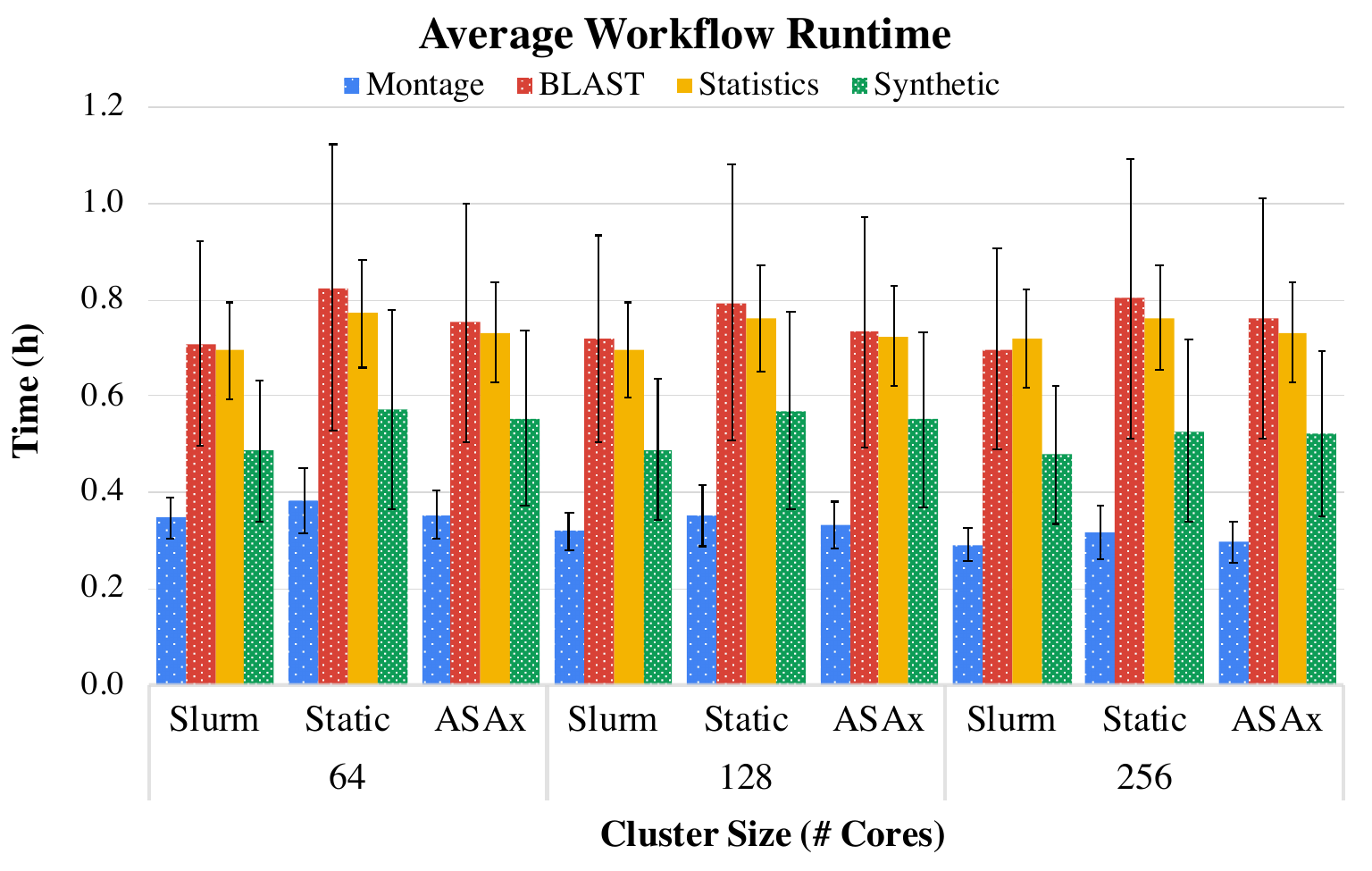}
        \label{fig:runtime}
    }
    \caption{Slurm, Static, and ASA$_X$ strategy results - Total (a) Queue makespan and (b) runtime of each cluster size (64, 128, and 256 cores) and scheduling strategy.}
    \label{fig:exps}
\end{figure*}

Table \ref{tab:results} shows workflows runtime as well as CPU and memory resource utilization, all normalized against Slurm. 
It can be seen that ASA$_X$ is close to Slurm regarding average total runtimes, with low overall overheads up to 10\%, reaching as low as 3\% increase for Montage.
A predictable, but notable achievement for both the Static setting and ASA$_X$ is the improved CPU utilization.
As only a specific fraction of resources are actually allocated in both strategies, the utilization ratio increases considerably.
For instance, when a job requests 1 CPU, the Static setting allocates 50\% of one CPU to one job and the other 50\% is co-allocated to other job.
For non CPU intensive workloads (Montage and Statistics), this strategy results in higher utilization ratios because such workloads consume less than the upper bound capacity, which is noticeable when comparing to Slurm.
For CPU intensive workloads such as BLAST, this strategy hurts performance, extending the workflow total runtime, most noticeably in the Static results.
In contrast, as ASA$_X$ learns overtime that co-scheduling jobs with either BLAST or Synthetic results in poor performance, its scheduling decisions become biased towards co-allocating Montage and Statistics workflows, but not BLAST or Synthetic. 
Also notable is the increased CPU utilization for the Statistics workflow, a non CPU intensive workload.
The only workload capable of utilizing most of the memory available in the system is the Synthetic workflow, which is both memory and CPU intensive.
This property makes ASA$_X$ avoid co-placing any other workload with Synthetic jobs, as memory is one of the key metrics in our decision trees (see previous section).

\begin{table}[!ht]
\caption{Workflows summary for Slurm, Static, and ASA$_X$ in three different cluster sizes. CPU Util. and Mem Util. represent resource utilization (\%) proportionally to total allocated capacity. 
Normalized averages over Slurm are shown below results for each cluster size. Acronyms: WF (Workflow), Stats (Statistics), and Synth. (Synthetic).
}
\label{tab:results}
\centering
\resizebox{\textwidth}{!}{%
\begin{tabular}{@{}r|r|rrr|rrr|rrrr@{}}
\toprule
\multicolumn{2}{c|}{} & \multicolumn{3}{c|}{\textbf{Slurm}} & \multicolumn{3}{c|}{\textbf{Static}} & \multicolumn{3}{c}{\textbf{ASA$_X$}} \\ \midrule
\multicolumn{1}{c|}{\textbf{\begin{tabular}[c]{@{}r@{}}Cluster\\Size\end{tabular}}} & \textbf{\begin{tabular}[c]{@{}r@{}}WF\end{tabular}} & \textbf{\begin{tabular}[c]{@{}r@{}}Runtime\\(h)\end{tabular}} & \textbf{\begin{tabular}[c]{@{}r@{}}CPU\\Util.\\(\%)\end{tabular}} & \textbf{\begin{tabular}[c]{@{}r@{}}Mem\\Util.\\(\%)\end{tabular}} & \textbf{\begin{tabular}[c]{@{}r@{}}Runtime\\(h)\end{tabular}} & \textbf{\begin{tabular}[c]{@{}r@{}}CPU\\Util.\\(\%)\end{tabular}} & \textbf{\begin{tabular}[c]{@{}r@{}}Mem\\Util.\\(\%)\end{tabular}} & \textbf{\begin{tabular}[c]{@{}r@{}}Runtime\\(h)\end{tabular}} & \textbf{\begin{tabular}[c]{@{}r@{}}CPU\\Util.\\(\%)\end{tabular}} & \textbf{\begin{tabular}[c]{@{}r@{}}Mem\\Util.\\(\%)\end{tabular}}
\\ \midrule
\multirow{4}{*}{\rotatebox[origin=c]{90}{\underline{64}}} & Montage & 0.35 $\pm$4\% & 45 $\pm$5 & 10 $\pm$5 & 0.38 $\pm$7\% & 94 $\pm$5 & 10 $\pm$5 & 0.35 $\pm$5\% & 95 $\pm$4 & 10 $\pm$4 \\
 & BLAST & 0.71 $\pm$21\% & 96 $\pm$4 & 44 $\pm$5 & 0.82 +30\% & 99 $\pm$1 & 44 $\pm$5 & 0.75 $\pm$25\% & 99 $\pm$1 & 44 $\pm$5 \\
 & Stats & 0.69 $\pm$10\% & 26 $\pm$2 & 5 $\pm$1 & 0.77 $\pm$30\% & 51 $\pm$3 & 6 $\pm$1 & 0.73 $\pm$11\% & 61 $\pm$4 & 5 $\pm$1 \\
 & Synth. & 0.49 $\pm$15\% & 99 $\pm$1 & 95 $\pm$3 & 0.57 $\pm$21\% & 99 $\pm$1 & 95 $\pm$1 & 0.56 $\pm$18\% & 99 $\pm$1 & 95 $\pm$1 \\
 \midrule
\multicolumn{2}{c|}{\textit{\textbf{\begin{tabular}[c]{@{}c@{}}Normalized \\ Average\end{tabular}}}} & - & - & - & \textit{+10\%} & \textit{+230\%} & \textit{0\%} & \textit{+3\%} & \textit{+227\%} & \textit{0\%} \\ \midrule
\multirow{4}{*}{\rotatebox[origin=c]{90}{\underline{128}}} & Montage & 0.32 $\pm$4\% & 33 $\pm$5 & 7 $\pm$3 & 0.35 $\pm$6\% & 73 $\pm$5 & 7 $\pm$3 & 0.33 $\pm$5\% & 75 $\pm$5 & 7 $\pm$3 \\
 & BLAST & 0.72 $\pm$22\% & 95 $\pm$3 & 45 $\pm$5 & 0.79 $\pm$29\% & 99 $\pm$1 & 44 $\pm$5 & 0.73 $\pm$24\% & 99 $\pm$1 & 44 $\pm$5 \\
 & Stats & 0.70 $\pm$10\% & 16 $\pm$2 & 1 $\pm$1 & 0.76 $\pm$11\% & 36 $\pm$6 & 2 $\pm$1 & 0.72 $\pm$11\% & 43 $\pm$7 & 1 $\pm$1 \\
 & Synth. & 0.50 $\pm$15\% & 99 $\pm$1 & 93 $\pm$4 & 0.57 $\pm$21\% & 99 $\pm$1 & 93 $\pm$3 & 0.55 $\pm$18\% & 99 $\pm$1 & 93 $\pm$4 \\
 \midrule
\multicolumn{2}{c|}{\textit{\textbf{\begin{tabular}[c]{@{}c@{}}Normalized \\ Average\end{tabular}}}} & - & - & - & \textit{+19\%} & \textit{+5\%} & \textit{0\%} & \textit{+9\%} & \textit{+5\%} & \textit{0\%} \\ \midrule
\multirow{4}{*}{\rotatebox[origin=c]{90}{\underline{256}}} & Montage & 0.29 $\pm$4\% & 21 $\pm$5 & 5 $\pm$3 & 0.32 $\pm$6\% & 55 $\pm$5 & 5 $\pm$3 & 0.30 $\pm$4\% & 51 $\pm$5 & 5 $\pm$3 \\
 & BLAST & 0.70 $\pm$21\% & 91 $\pm$3 & 42 $\pm$6 & 0.80 $\pm$29\% & 99 $\pm$1 & 43 $\pm$6 & 0.76 $\pm$25\% & 99 $\pm$1 & 43 $\pm$6 \\
 & Stats & 0.72 $\pm$10\% & 10 $\pm$2 & 1 $\pm$1 & 0.76 $\pm$11\% & 26 $\pm$2 & 1 $\pm$1 & 0.73 $\pm$11\% & 86 $\pm$6 & 1 $\pm$1 \\
 & Synth. & 0.48 $\pm$14\% & 99 $\pm$1 & 92 $\pm$2 & 0.53 $\pm$19\% & 99 $\pm$1 & 91 $\pm$2 & 0.52 $\pm$17\% & 99 $\pm$1 & 93 $\pm$1 \\
 \midrule
\multicolumn{2}{c|}{\textit{\textbf{\begin{tabular}[c]{@{}c@{}}Normalized \\ Average\end{tabular}}}} & - & - & - & \textit{+10\%} & \textit{+36\%} & \textit{0\%} & \textit{+5\%} & \textit{+50\%} & \textit{+1\%} \\ 
 \bottomrule
\end{tabular}
}
\end{table}

\begin{table}
\caption{Slurm, Static, and ASA$_X$ - Average results for three strategies in each cluster size.}
\label{tab:summary}
\centering
\resizebox{0.7\textwidth}{!}{%
\begin{tabular}{@{}rrr|rrrrr@{}}
\toprule

\multicolumn{1}{c}{\textbf{}} & \textbf{\begin{tabular}[c]{@{}r@{}}Cluster\\Size\end{tabular}} & \textbf{\begin{tabular}[c]{@{}r@{}}Cluster\\Load\end{tabular}} & \textbf{\begin{tabular}[c]{@{}r@{}}Waiting\\Time (h)\end{tabular}} & \textbf{\begin{tabular}[c]{@{}r@{}}CPU \\Util. (\%)\end{tabular}} & \textbf{\begin{tabular}[c]{@{}r@{}}Response\\Time (h) \end{tabular}} \\ \midrule
\multirow{3}{*}{\rotatebox[origin=c]{90}{\textbf{\underline{Slurm}}}} & 64 & 8x & 3.5$\pm$1\% & 53$\pm$5 & 4.4$\pm$1\% \\
 & 128 & 4x & 1.5$\pm$1\% & 45$\pm$5 & 2.4$\pm$1\% \\ 
 & 256 & 2x & 0.5$\pm$1\% & 32$\pm$6 & 1.4$\pm$1\% \\ \midrule
\multirow{3}{*}{\rotatebox[origin=c]{90}{\underline{\textbf{Static}}}} & 64 & 8x & 1.7$\pm$1\% & 90$\pm$5 & 4.8$\pm$1\% \\
 & 128 & 4x & 0.8$\pm$1\% & 92$\pm$3 & 2.8$\pm$1\% \\
 & 256 & 2x & 0.3$\pm$1\% & 89$\pm$5 & 2.0$\pm$1\% \\ \midrule
\multirow{3}{*}{\rotatebox[origin=c]{90}{\textbf{\underline{ASA$_X$}}}} & 64 & 8x & 1.8$\pm$5\% & 82$\pm$7 & 3.5$\pm$3\% \\
 & 128 & 4x & 1.0$\pm$8\% & 84$\pm$5 & 2.0$\pm$2\% \\
 & 256 & 2x & 0.3$\pm$7\% & 83$\pm$4 & 0.9$\pm$2\% \\ \bottomrule
\end{tabular}
}
\end{table}

\subsubsection{Aggregated Queue and Cluster Metrics.}

Whereas total runtime results are important for the individual users, aggregated cluster metrics such as queue waiting times matter for cluster operators. 
The latter metrics are summarized in Table \ref{tab:summary} that shows average waiting times (h), cluster CPU utilization (\%), and response times (h).
The key point in Table \ref{tab:summary} relates to both average response time and queue waiting time.
Queue waiting times are reduced by as much as 50\% in both Static and ASA$_X$ compared to Slurm.
ASA$_X$ has 55\% lower response time (256 cores) than Slurm or Static, which shows that it makes co-scheduling decisions that -- for the proposed workload -- result in fast executions. 
Similarly to the Static setting, ASA$_X$ also increases cluster (CPU) utilization by up to 59\%, but with the advantage of low performance degradation while also reducing queue waiting times considerably. 
This also happens for low cluster loads, as can be seen in the x2 workload case, with 51\% average waiting time reductions.
Conversely, Slurm decreases CPU utilization as the load decreases, and it is thus not able to improve response time.
\section{Discussion} \label{sec:disc}

The evaluation demonstrates how ASA$_X$ combines application profiling with an assertive learning algorithm to simultaneously improve resource usage and cluster throughput, with direct reductions on the workload makespan.
By combining an intuitive but powerful abstraction (decision tree experts) and an application agnostic (happiness) performance, ASA$_X$ efficiently co-schedules jobs and improves cluster throughput. The overall performance degradation experienced by ASA$_X$ (10\% average runtime slowdown) is negligible when compared to its benefits, in particular as HPC users are known to overestimate walltimes~\cite{uchronski2018user, cirne2001comprehensive, rocchetti2016penalty}.
When compared to Slurm, a very common batch system used in HPC infrastructures, ASA$_X$ reduces  average job response times by up to 10\%, enabling HPC users to achieve faster time to results.
Notably, these improvements are achieved also for lightly loaded clusters, where the co-scheduling actions optimize cluster resources and allows them to be shared with other jobs.

It is important to note that current HPC scheduling strategies, specially those based on backfilling, aim to maximize resource allocation at a coarse granularity level.
It is generally assumed this is good policy, as it guarantees high allocation ratios without interfering with users' workflows and thus achieves predictable job performance.
However, such strategies do not take advantage of the characteristics of modern dynamic workflows,
namely adaptability to changed resource allocation, faults, and even on input accuracy.
Nor does backfilling strategies take advantage of modern operating system capabilities such as
fine-grained processes scheduling and access control such as Linux \cite{cgroup}, available in userspace. Notably in HPC, even classical rigid jobs may suffer from traditional space-sharing policies, as the reservation based model does not take into consideration how resources are utilized in runtime. 
This is due to historical reasons, mostly related to CPU scarcity and the need for application performance consistency and SLO guarantees. 
In contrast, by leveraging the dynamic nature of modern workflows, and by using well fine-grained resource control mechanisms such as cgroups, ASA$_X$ is able to improve the cluster utilization without hurting jobs performance. 
Although applicable across a wide range of HPC environments, such scheduling features are particularly suitable for rackscale systems like Numascale that offers close to infinite vertical scaling of applications. 
In such systems, scheduling mechanisms like ASA$_X$ are key as the default Linux Completely Fair Scheduler (CFS) -- beyond cache-locality -- does not understand fine-grained job requirements, such as walltimes limits and data locality.

Given one cluster workload, the goals of ASA$_X$ are higher cluster throughput and increased resource utilization. A constraint for these goals is the performance degradation caused by resource sharing and collocation, as measured by applications' completion time. In our experiments, the same queue workload is submitted to three resource capacity scenarios and cluster sizes that creates low, medium, and high demand, respectively. This captures the main characteristics that many HPC clusters face as the cluster size and workload directly affect how resources should be scheduled, in particular given the goal to maximize resource utilization without affecting application runtime.
%
%
%
As shown in Table 2 (Response Time), and due to collocation and sharing of resources, ASA$_X$ impacts all application's completion time. This is mitigated by decisions the algorithm takes at runtime, which uses the applications state to evaluate the impacts of each collocation. However, it must also be noted that no hard deadlines are set for jobs' wall-clock times. This allows comparison of ASA$_X$ to worst case scenarios, and enables evaluation of the overhead caused by resource sharing due to collocation.
%
Collocation has neglectable total runtime impacts (up to 5\% in our experiments) and users are known to request more time than needed when submitting their jobs. Some predictable noisy interference is bound to happen at runtime, and ASA$_X$'s goal is simply to model this expected behavior.
By learning from mistakes, ASAx avoids collocating two applications that may negatively impact the scheduler's loss function as described by the decision trees, evaluated at each time that the scheduler has to make a collocation decision.

Finally, the proposed \textit{happiness} metric (Eq. \ref{eq.happiness}) can enable ASA$_X$ to  mitigate and control the possible performance degradation due to bad collocation decisions.
Together with the decision trees, this characteristic is specially useful for queues that share nodes among two or more jobs, e.g. debugging queues used during scientific applications development stages.
It is important to note though, that the happiness metric assumes all tasks in a single job are sufficiently homogeneous, i.e., need approximately the same amount of time to complete execution.
This is however the common case in most stages of a scientific workflow, as in-stage heterogeneity leads to inefficient parallelization and poor performance. 
As such, monitoring $Hp(t)$ before and after the co-scheduling of a job can also be useful to understand if such a decision actually is likely to succeed. However, this is outside the scope of this paper but can be seen as an natural extension to our co-scheduling algorithm as such feature can enable ASA$_X$ to foresee performance degradation in running applications.
This can ultimately enable ASA$_X$ to optimize its co-scheduling actions already before collocation,
which further protects performance of collocated jobs. 
Another possible extensions is to generalize the ASA$_X$ allocation capacity to accommodate more than two jobs in the same node simultaneously as long as they all impose limited performance degradation on one other.
\section{Related Work} \label{sec:related}
The scale and processing capacities of modern multi-core and heterogeneous servers create new HPC resource management opportunities.
With increasing concerns in datacenter density, idle capacity should be utilized as efficiently as possible~\cite{stevens2020ai}. 
However, the low effective utilization of computing resources in datacenters is explained through certain important requirements. 
For instance, to guarantee user Quality-of-Services (QoS), datacenter operators adopt a worst-case policy in resource allocations~\cite{yang2013bubble}, while in HPC, users expect stable performance guarantees.
To reach QoS guarantees, users often overestimate resource requests and total execution times (wall-clock), further degrading datacenter throughput.
Runtime Service Level Objectives (SLOs) are achieved through more efficient job scheduling, but this is rarely deployed in modern HPC infrastructures due to fear of performance degradation when sharing resources. 
Most of HPC batch systems use resource reservation with backfilling strategies, and rely on users to provide application resource and walltime requirements~\cite{slurm02, torque06}.
This model is based on the assumption that jobs fully utilize the allocated capacity constantly throughout their lifespan.
This is needed for some purposes such as cache optimizations, debugging, and is thus important during application development and testing.
However, more often than not, jobs utilize less than this upper capacity~\cite{tirmazi2020borg}, which is rarely the case, specially at early stages of development.
Some proposals \cite{ahn2014flux, domeniconi2019cush} aim to extend this traditional model, with impacts need to be studied in more depth.
Public cloud datacenters, on the other hand, offer alternatives for running HPC workloads, such as Kubernetes \cite{burns2016borg}, and Mesos \cite{2011mesos}.
Their resource management policies are centered around low latency scheduling, offering fairness and resource negotiation \cite{2011drf}.
However, these systems lack main capabilities such as resource reservation, gang scheduling, and also assume that applications fully utilize allocated resources \cite{tirmazi2020borg, feitelson1996packing}.

Tackling the utilization problem from a different angle, stochastic schedulers have been proposed as solutions to overcome user resource overestimates~\cite{gainaru2019speculative}. 
ASA$_X$ can be used as an extension to this class of schedulers, because it offers a new scheduling abstraction through its experts, which can model stochastic applications as well. 
Similarly to ASA$_X$, Deep RL schedulers have been proposed \cite{zhang2019rlscheduler, li2020dynamic, mao2019learning, domeniconi2019cush}, though they focus either on dynamic applications, or on rigid-jobs.
As previously discussed, an increasingly diverse set of applications are flooding new HPC infrastructure realizations.
\cite{patel2020clite} uses Bayesian optimization to capture performance models between low latency and batch jobs.
It differs from our experts approach in that ASA$_X$ uses decision trees to assess such relationship, and \cite{patel2020clite} targets only two types of applications.
Machine Learning (ML) models have been proposed to take advantage of the large datasets already available in current datacenters~\cite{carastan2017obtaining}.
For instance, \cite{moradi2020upredict} repeatedly runs micro-benchmarks in order to build a predictive model for the relationship between an application and resource contention.
In contrast to that type of ML, RL algorithms such as the one used in ASA$_X$ and in \cite{zhang2019rlscheduler} do not require historical training data to optimize scheduling.
As mentioned in previous sections, ASA$_X$ is a stateful extension of \cite{asa}, where the main difference is that ASA$_X$ can incorporate previous decisions in a RL approach.
In \cite{thamsenhugo}, the authors combine offline job classification with online RL to improve collocations. This approach can accelerate convergence, although it might have complex consequences when new unknown jobs do not fit into the initial classification.
Similarly to ASA$_X$, the Cognitive ScHeduler (CuSH) \cite{domeniconi2019cush} decouples job selection from the policy selection problem and also handles heterogeneity.
Finally, as an example of scheduling policy optimization, \cite{zhang2019rlscheduler} selects among several policies to adapt the scheduling decisions.
This approach is similar to our concept of a forest of decision trees (experts), although their work do not consider co-scheduling to target queue waiting time minimization combined with improved cluster utilization.
\section{Conclusion} \label{sec:conclusion}

Since mainframes, batch scheduling has been an area of research, and even more since time-sharing schedulers were first proposed.
However, today's HPC schedulers face a diverse set of workloads and highly dynamic requirements, such as low latency and streaming workflows generated by data-intensive applications.
These workflows are characterized by supporting many different features, such as system faults, approximate output, and resource adaptability.
However, current HPC jobs do not fully utilize the capacity provided by these high-end infrastructures, impacting datacenter operational costs, besides hurting user experience due to long waiting times.
To mitigate these, in this paper we propose a HPC co-scheduler (ASA$_X$) that uses a novel and convergence proven reinforcement learning algorithm.
By analytically describing a cluster's policy through efficient decision trees, our co-scheduler is able to optimize job collocation by profiling applications and understanding how much of ongoing allocations can be safely reallocated to other jobs.
Through real cluster experiments, we show that ASA$_X$ is able to improve cluster CPU utilization by as much as 51\% while also reducing response time and queue waiting times, hence improving overall datacenter throughput, with little impacts on job runtime (up to 10\%, but 5\% slower on average).
Together with the architecture, our algorithm forms the basis for efficient resource sharing and an application-aware co-scheduler for improved HPC scheduling with minimal performance degradation.
\appendix{}
\section{Convergence of ASA$_{\mbox{X}}$} \label{sec:appendix}

In here we explain how we can statistically bound the accumulated loss and create a strategy to take the best available actions in a given environment.
For doing so, we define the \textit{excess risk}, which estimates how "risky" taking an action can be.
It is defined here as 
\begin{equation}
\begin{split}
	E_t  
	= 
	\sum_{s=1}^t \sum_{j=1}^{n_s} \left( \sum_{i=1}^n \alpha_{s-1,i}\pv_i(\xv_{sj}) \ell\left(a_{sj}\right) \right)
	- \\
	\min_\ast \sum_{s=1}^t \sum_{j=1}^{n_s} \pv_\ast(\xv_{sj}) \ell\left(a_{sj}\right),
	\label{eq.excess}
\end{split}
\end{equation}
where $\xv_{sj}$ denotes the decision tree states of the $j$-th case in the $s$-th round, and where
$a_{sj}$ is the action taken at this case.

Theorem \ref{th.asax} shows that the excess risk for Algorithm \ref{alg.asax} is bound as follows.
\begin{Theorem}
	Let $\{\gamma_t>0\}_t$ be a non-increasing sequence.
	The excess risk $E_t$ after $t$ rounds is then bound by
	\begin{equation}
		E_t \leq \gamma_t ^{-1}\left(\ln n + \frac{1}{2}\sum_{s=1}^t  \gamma_s^2\right).
		\label{eq.asax}
	\end{equation}
	\label{th.asax}
\end{Theorem}
\begin{proof}
	Let $a_{tj}$ denote the action taken in round $t$ at the $j$th case, and let  $n_t$ denote the number of cases in round $t$.
	Similarly, let $\xv_{tj}$ denote the state for this case, and let $\pv_i(\xv_{tj})$ denote the distribution over the $a$ actions as proposed by the $i$th expert.
	Let $\ell_{tj}(a)$ denote the (not necessarily observed) loss of action $a$ as achieved on the $tj$th case.
	
	Define the variable $Z_t$ as 
	\begin{equation}
		Z_t = \sum_{i=1}^n \exp\left( - \sum_{s=1}^t \gamma_s \sum_{j=1}^{n_s} \pv_i(\xv_{sj}) \ell_{sj}\left(a_{sj}\right) \right).
		\label{eq.asax.}
	\end{equation}
	Then 
	\begin{equation}
		\sum_{s=1}^t\ln \frac{Z_s}{Z_{s-1}} 
		= 
		\ln Z_t - \ln Z_0.
		\label{eq.asax.}
	\end{equation}
	Moreover
	\begin{equation}
	\begin{split}
		\ln Z_t 
		= 
		\ln \sum_{i=1}^n \exp\left( - \sum_{s=1}^t \gamma_s \sum_{j=1}^{n_s} \pv_i(\xv_{sj}) \ell_{sj}\left(a_{sj}\right) \right)
		\geq \\	
		- \sum_{s=1}^t \gamma_s \sum_{j=1}^{n_s} \pv_\ast(\xv_{sj}) \ell_{sj}\left(a_{sj}\right),
		\label{eq.asax.}
		\end{split}
	\end{equation}
	for any expert $\ast \in\{1,\dots,n\}$.
	Conversely, we have
	\begin{multline}
		\ln \frac{Z_t}{Z_{t-1}} 
		= 
		\ln \frac{\sum_{i=1}^n \exp\left( -  \sum_{s=1}^t \gamma_s \sum_{j=1}^{n_s} \pv_i(\xv_{sj}) \ell_{sj}\left(a_{sj}\right)  \right)}
			{\sum_{i=1}^n \exp\left( - \sum_{s=1}^{t-1} \gamma_s \sum_{j=1}^{n_s} \pv_i(\xv_{sj}) \ell_{sj}\left(a_{sj}\right)  \right)} \\
		=
		\ln \sum_{i=1}^n \alpha_{t-1,i} \exp\left( - \gamma_t \sum_{j=1}^{n_t} \pv_i(\xv_{tj}) \ell_{tj}\left(a_{tj}\right)  \right).
		\label{eq.asax.}
	\end{multline}
	Application of Hoeffding's lemma gives then
	\begin{equation}
		\ln \frac{Z_s}{Z_{s-1}} 
		\leq 
		- \gamma_s \sum_{j=1}^{n_s} \left( \sum_{i=1}^n \alpha_{s-1,i}\pv_i(\xv_{sj}) \ell_{sj}\left(a_{sj}\right) \right)  + \frac{4\gamma_s^2}{8},
		\label{eq.asax.}
	\end{equation}
	using the construction that $\max_i\sum_{j=1}^{n_t} \pv_i(\xv_{sj}) \ell_{sj}\left(a_{sj}\right)\leq 1$. 
	Reshuffling terms gives then
	\begin{equation}
	\begin{split}
		 \sum_{s=1}^t \gamma_s \sum_{j=1}^{n_s} \left( \sum_{i=1}^n \alpha_{s-1,i}\pv_i(\xv_{sj}) \ell_{sj}\left(a_{sj}\right) \right)
		- \\
		 \sum_{s=1}^t \gamma_s \sum_{j=1}^{n_s} \pv_\ast(\xv_{sj}) \ell_{sj}\left(a_{sj}\right) 
		\leq
		\ln n + \frac{1}{2}\sum_{s=1}^t  \gamma_s^2.
		\label{eq.asax.}
		\end{split}
	\end{equation}
	Application of Abel's second inequality~\cite{abelineq} gives the result.
\end{proof}

%
%
\bibliographystyle{splncs04}
\bibliography{references}

\end{document}